%% file: manuscript-double.tex
\documentclass[journal,twocolumn,10pt,twoside]{IEEEtran}

\usepackage{amsfonts}
\usepackage{amssymb}
\usepackage{amsmath,amssymb,epsfig}
\usepackage{amsthm}
\usepackage{rawfonts,graphicx,amsfonts,amssymb}
\usepackage{array}
\usepackage{mathrsfs}
\usepackage{cite}
\usepackage{algorithm}
\usepackage{algorithmic}
\usepackage{setspace}
\usepackage{psfrag,pstricks}

\newtheorem{theorem}{Theorem}
\newtheorem{proposition}{Proposition}

\IEEEoverridecommandlockouts
\ifCLASSINFOpdf \else \fi

\begin{document}

\title{Slow Adaptive OFDMA Systems\\
Through Chance Constrained Programming}

\author{William Wei-Liang Li,~\IEEEmembership{Student~Member,~IEEE},
Ying Jun (Angela) Zhang,~\IEEEmembership{Member,~IEEE}, \\Anthony
Man-Cho So, and Moe Z. Win,~\IEEEmembership{Fellow,~IEEE}

\thanks{Copyright~\copyright~2010 IEEE. Personal use of this material is permitted. However, permission to use this material for any other purposes must be obtained from the IEEE by sending a request to pubs-permissions@ieee.org.

Manuscript received July 01, 2009; revised October 28, 2009 and February 09, 2010; accepted February 15, 2010. This research was supported, in part, by the Competitive Earmarked Research Grant (Project numbers 418707 and 419509) established under the University Grant Committee of Hong Kong, Project \#MMT-p2-09 of the Shun Hing Institute of Advanced Engineering, the Chinese University of Hong Kong, the National Science Foundation under Grants ECCS-0636519 and ECCS-0901034, the Office of Naval Research Presidential Early Career Award for Scientists and Engineers (PECASE) N00014-09-1-0435, and the MIT Institute for Soldier Nanotechnologies. The associate editor coordinating the review of this manuscript and approving it for publication was Dr. Walid Hachem.

W. W.-L. Li is with the Department of Information Engineering, the Chinese University of Hong Kong, Hong Kong (wlli@ie.cuhk.edu.hk). 

Y. J. Zhang is with the Department of Information Engineering and the Shun Hing Institute of Advanced Engineering, the Chinese University of Hong Kong, Hong Kong (yjzhang@ie.cuhk.edu.hk). 

A. M.-C. So is with the Department of Systems Engineering and Engineering Management and the Shun Hing Institute of Advanced Engineering, the Chinese University of Hong Kong, Hong Kong (manchoso@se.cuhk.edu.hk). 

M. Z. Win is with the Laboratory for
Information \& Decision Systems (LIDS), Massachusetts Institute of
Technology, MA, USA (moewin@mit.edu).

Digital Object Identifier XXX/XXX
}}

\markboth{IEEE Transactions on Signal Processing, Vol.~X, No.~X, XXX 2010}{Li \MakeLowercase{\textit{et al.}}: Slow Adaptive OFDMA Systems through Chance Constrained Programming}
\maketitle

\begin{abstract}
Adaptive OFDMA has recently been recognized as a promising technique
for providing high spectral efficiency in future broadband wireless
systems. The research over the last decade on adaptive OFDMA systems
has focused on adapting the allocation of radio resources, such as
subcarriers and power, to the instantaneous channel conditions of
all users. However, such ``fast'' adaptation requires high
computational complexity and excessive signaling overhead. This
hinders the deployment of adaptive OFDMA systems worldwide. This
paper proposes a slow adaptive OFDMA scheme, in which the subcarrier
allocation is updated on a much slower timescale than that of the
fluctuation of instantaneous channel conditions. Meanwhile, the data
rate requirements of individual users are accommodated on the fast
timescale with high probability, thereby meeting the requirements
except occasional outage. Such an objective has a natural chance
constrained programming formulation, which is known to be
intractable. To circumvent this difficulty, we formulate safe
tractable constraints for the problem based on recent advances in
chance constrained programming. We then develop a polynomial-time
algorithm for computing an optimal solution to the reformulated
problem. Our results show that the proposed slow adaptation scheme
drastically reduces both computational cost and control signaling
overhead when compared with the conventional fast adaptive OFDMA.
Our work can be viewed as an initial attempt to apply the chance
constrained programming methodology to wireless system designs.
Given that most wireless systems can tolerate an occasional dip in
the quality of service, we hope that the proposed methodology will
find further applications in wireless communications.
\end{abstract}


\IEEEpeerreviewmaketitle

\singlespacing

\section{Introduction}\label{intro}
\IEEEPARstart{F}{uture} wireless systems will face a growing demand for broadband and multimedia services. Orthogonal frequency division multiplexing (OFDM) is a leading technology to meet this demand due to its ability to mitigate wireless channel impairments. The inherent
multicarrier nature of OFDM facilitates flexible use of subcarriers
to significantly enhance system capacity. Adaptive subcarrier
allocation, recently referred to as adaptive orthogonal frequency
division multiple access (OFDMA) \cite{WonCheLetMur:99,ZhaLet:04}, has been considered as a primary contender in next-generation wireless
standards, such as IEEE802.16 WiMAX \cite{802.16e:Wimax} and 3GPP-LTE
\cite{3GPP:LTE}.

In the existing literature, adaptive OFDMA exploits time, frequency,
and multiuser diversity by quickly adapting subcarrier allocation
(SCA) to the instantaneous channel state information (CSI) of all
users. Such ``fast'' adaptation suffers from high computational
complexity, since an optimization problem required for adaptation
has to be solved by the base station (BS) every time the channel
changes. Considering the fact that wireless channel fading can vary
quickly (e.g., at the order of milli-seconds in wireless cellular
system), the implementation of fast adaptive OFDMA becomes
infeasible for practical systems, even when the number of users is
small. Recent work on reducing complexity of fast adaptive OFDMA
includes \cite{WonEva:08,MarGiaDigRam:08}, etc. Moreover, fast adaptive
OFDMA requires frequent signaling between the BS and mobile users in
order to inform the users of their latest allocation decisions. The
overhead thus incurred is likely to negate the performance gain
obtained by the fast adaptation schemes. To date, high computational
cost and high control signaling overhead are the major hurdles that
prevent adaptive OFDMA from being deployed in practical systems.

We consider a slow adaptive OFDMA scheme, which is motivated by
\cite{ConWinChi:07}, to address the aforementioned problem. In contrast
to the common belief that radio resource allocation should be
readapted once the instantaneous channel conditions change, the
proposed scheme updates the SCA on a much slower timescale than that
of channel fluctuation. Specifically, the allocation decisions are
fixed for the duration of an adaptation window, which spans the
length of many coherence times. By doing so, computational cost and
control signaling overhead can be dramatically reduced. However,
this implies that channel conditions over the adaptation window are
uncertain at the decision time, thus presenting a new challenge in the design of slow adaptive OFDMA schemes. An important question is how
to find a valid allocation decision that remains optimal and
feasible for the entire adaptation window. Such a problem can be
formulated as a stochastic programming problem, where the channel
coefficients are random rather than deterministic.

Slow adaptation schemes have recently been studied in other contexts
such as slow rate adaptation \cite{ConWinChi:07,LiKis:08} and slow power
allocation \cite{QueShiWin:07}. Therein, adaptation decisions are made
solely based on the long-term average channel conditions instead of
fast channel fading. Specifically, random channel parameters are
replaced by their mean values, resulting in a deterministic rather
than stochastic optimization problem. By doing so,
quality-of-service (QoS) can only be guaranteed in a long-term
average sense, since the short-term fluctuation of the channel is
not considered in the problem formulation. With the increasing
popularity of wireless multimedia applications, however, there will
be more and more inelastic traffic that require a guarantee on the
minimum short-term data rate. As such, slow adaptation schemes based
on average channel conditions cannot provide a satisfactory QoS.

On another front, robust optimization methodology can be applied to
meet the short-term QoS. For example, robust optimization method was
applied in \cite{QueShiWin:07,QueWinChi:10,LiZhaWin:09} to find a solution that is feasible for the entire uncertainty set of channel conditions, i.e., to guarantee the instantaneous data rate requirements regardless of the channel realization. Needless to say, the resource allocation
solutions obtained via such an approach are overly conservative. In
practice, the worst-case channel gain can approach zero in deep
fading, and thus the resource allocation problem can easily become
infeasible. Even if the problem is feasible, the resource
utilization is inefficient as most system resources must be
dedicated to provide guarantees for the worst-case scenarios.

Fortunately, most inelastic traffic such as that from multimedia applications can tolerate an occasional dip in the instantaneous data rate
without compromising QoS. This presents an opportunity to enhance
the system performance. In particular, we employ chance constrained
programming techniques by imposing probabilistic constraints on user
QoS. Although this formulation captures the essence of the problem,
chance constrained programs are known to be computationally
intractable except for a few special cases \cite{BirLou:B97}. In
general, such programs are difficult to solve as their feasible sets
are often non-convex. In fact, finding feasible solutions to a
generic chance constrained program is itself a challenging research
problem in the Operations Research community. It is partly due to
this reason that the chance constrained programming methodology is
seldom pursued in the design of wireless systems.

In this paper, we propose a slow adaptive OFDMA scheme that aims at
maximizing the long-term system throughput while satisfying with high
probability the short-term data rate requirements. The key
contributions of this paper are as follows:
\begin{itemize}
\item We design the slow adaptive OFDMA system based on chance
constrained programming techniques. Our formulation guarantees the
short-term data rate requirements of individual users except in rare
occasions. To the best of our knowledge, this is the first work that uses
chance constrained programming in the context of resource allocation
in wireless systems.
\item We exploit the special structure of the probabilistic constraints
in our problem to construct safe tractable constraints (STC) based
on recent advances in the chance constrained programming literature.
\item We design an interior-point algorithm that is tailored for the slow adaptive OFDMA problem, since the formulation with STC, although convex, cannot be trivially solved using off-the-shelf optimization software. Our algorithm can efficiently compute an optimal solution to the problem with STC in polynomial time.
\end{itemize}

The rest of the paper is organized as follows. In Section
\ref{model}, we discuss the system model and problem formulation. An
STC is introduced in Section \ref{approx} to solve the original
chance constrained program. An efficient tailor-made algorithm for
solving the approximate problem is then proposed in Section \ref{algo}.
In Section \ref{reduct}, we reduce the problem size based on some
practical assumptions, and show that the revised problem can be solved by the proposed algorithm with much lower complexity. In Section
\ref{simu}, the performance of the slow adaptive OFDMA system is
investigated through extensive simulations. Finally, the paper is
concluded in Section \ref{conclude}.

\section{System Model and Problem Formulation}\label{model}
This paper considers a single-cell multiuser OFDM system with $K$
users and $N$ subcarriers. We assume that the instantaneous channel
coefficients of user $k$ and subcarrier $n$ are described by complex
Gaussian\footnote{Although the techniques used in this paper are
applicable to any fading distribution, we shall prescribe to a
particular distribution of fading channels for illustrative
purposes.} random variables
$h_{k,n}^{(t)}\sim\mathcal{CN}(0,\sigma_k^2)$,
independent\footnote{The case when frequency correlations exist
among subcarriers will be discussed in Section \ref{simu}.} in both
$n$ and $k$. The parameter $\sigma_k$ can be used to model the
long-term average channel gain as
$\sigma_k=\left(\frac{d_k}{d_0}\right)^{-\gamma}\!\!\cdot s_k$, where $d_k$ is
the distance between the BS and subscriber $k$, $d_0$ is the
reference distance, $\gamma$ is the amplitude path-loss exponent and
$s_k$ characterizes the shadowing effect. Hence, the channel gain
$g_{k,n}^{(t)}=\big|h_{k,n}^{(t)}\big|^2$ is an exponential random
variable with probability density function (PDF) given by
\begin{equation}\label{expon}
f_{g_{k,n}}(\xi)=\frac{1}{\sigma_k}\exp\left(-\frac{\xi}{\sigma_k}\right).
\end{equation}
The transmission rate of user $k$ on subcarrier $n$ at time $t$ is
given by
$$
r^{(t)}_{k,n}=W\log_2\left(1+\frac{p_tg_{k,n}^{(t)}}{\Gamma
N_0}\right),
$$
where $p_t$ is the transmission power of a subcarrier,
$g_{k,n}^{(t)}$ is the channel gain at time $t$, $W$ is the
bandwidth of a subcarrier, $N_0$ is the power spectral density of
Gaussian noise, and $\Gamma$ is the capacity gap that is related to
the target bit error rate (BER) and coding-modulation schemes.

In traditional fast adaptive OFDMA systems, SCA decisions are made
based on instantaneous channel conditions in order to maximize the
system throughput. As depicted in Fig. \ref{fg_sys.mdl}a, SCA is
performed at the beginning of each time slot, where the duration of
the \emph{slot} is no larger than the coherence time of the channel.
Denoting by $x_{k,n}^{(t)}$ the fraction of airtime assigned to user
$k$ on subcarrier $n$, fast adaptive OFDMA solves at each time slot
$t$ the following linear programming problem:
\begin{align}
\mathcal{P}_\text{fast}:\quad\max_{x_{k,n}^{(t)}} &\quad\sum\limits_{k=1}^K\sum\limits_{n=1}^N{x^{(t)}_{k,n}r^{(t)}_{k,n}} \label{obj-f}\\
\text{s.t.}
&\quad\sum\limits_{n=1}^N{x_{k,n}^{(t)}r^{(t)}_{k,n}}\geq q_k, \quad\forall k \label{const-f}\\
&\quad\sum\limits_{k=1}^K x_{k,n}^{(t)} \le 1, \quad\forall n \notag\\
&\quad x_{k,n}^{(t)}\geq 0, \quad\forall k,n, \notag
\end{align}
where the objective function in \eqref{obj-f} represents the total
system throughput at time $t$, and \eqref{const-f} represents the data
rate constraint of user $k$ at time $t$ with $q_k$ denoting the
minimum required data rate. We assume that $q_k$ is known by the BS
and can be different for each user $k$. Since $g_{k,n}^{(t)}$ (and
hence $r^{(t)}_{k,n}$) varies on the order of coherence time, one
has to solve the Problem $\mathcal{P}_\text{fast}$ at the beginning
of every time slot $t$ to obtain SCA decisions. Thus, the above fast
adaptive OFDMA scheme is extremely costly in practice.

\begin{figure}[t]
\begin{center}
\includegraphics [height=5cm]{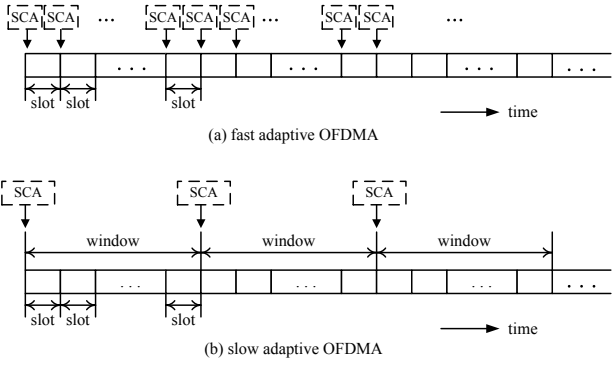}
\caption{Adaptation timescales of fast and slow adaptive OFDMA
system (SCA = SubCarrier Allocation).} \label{fg_sys.mdl}
\end{center}
\end{figure}

In contrast to fast adaptation schemes, we propose a slow adaptation
scheme in which SCA is updated only every \emph{adaptation window}
of length $T$. More precisely, SCA decision is made at the beginning
of each adaptation window as depicted in Fig. \ref{fg_sys.mdl}b, and
the allocation remains unchanged till the next window. We consider
the duration $T$ of a window to be large compared with that of fast
fading fluctuation so that the channel fading process over the
window is ergodic; but small compared with the large-scale channel
variation so that path-loss and shadowing are considered to be fixed
in each window. Unlike fast adaptive systems that require the exact
CSI to perform SCA, slow adaptive OFDMA systems rely only on the
distributional information of channel fading and make an SCA
decision for each window.

Let $x_{k,n}\in[0,1]$ denote the SCA for a given adaptation
window\footnote{It is practical to assume $x_{k,n}$ as a real number
in slow adaptive OFDMA. Since the data transmitted during each
window consists of a large mount of OFDM symbols, the time-sharing
factor $x_{k,n}$ can be mapped into the ratio of OFDM symbols
assigned to user $k$ for transmission on subcarrier $n$.}. Then, the
time-average throughput of user $k$ during the window becomes
\begin{equation*}
	\bar{b}_k=\sum_{n=1}^Nx_{k,n}\bar{r}_{k,n},
\end{equation*}
where
\begin{equation*}
\hspace{-0.35cm}
\bar{r}_{k,n}=\frac{1}{T}\int_T r_{k,n}^{(t)}dt	
\end{equation*}
is the time-average data rate of user $k$ on subcarrier $n$ during
the adaptation window. The time-average system throughput is given by
\begin{equation*}
	\bar{b}=\sum_{k=1}^K\sum_{n=1}^Nx_{k,n}\bar{r}_{k,n}.	
\end{equation*}
Now, suppose that each user has a short-term data rate requirement
$q_k$ defined on each time slot. If
$\sum_{n=1}^Nx_{k,n}r_{k,n}^{(t)}< q_k$, then we say that a rate
outage occurs for user $k$ at time slot $t$, and the probability of
rate outage for user $k$ during the window $[t_0,t_0+T]$ is defined
as
$$
P_k^{\text{out}}\triangleq\text{Pr}\left\{\sum_{n=1}^N
x_{k,n}r_{k,n}^{(t)} < q_k\right\}, \quad\forall t\in[t_0,t_0+T],
$$
where $t_0$ is the beginning time of the window.

Inelastic applications, such as voice and multimedia, that are
concerned with short-term QoS can often tolerate an occasional dip
in the instantaneous data rate. In fact, most applications can run
smoothly as long as the short-term data rate requirement is
satisfied with sufficiently high probability. With the above
considerations, we formulate the slow adaptive OFDMA problem as
follows:
\begin{align}
\mathcal{P}_\text{slow}:\,~\max_{x_{k,n}} &
\quad\sum_{k=1}^K\sum_{n=1}^N
x_{k,n}\mathbb{E}\left\{r_{k,n}^{(t)}\right\} \label{obj1}\\
\text{s.t.} & \quad\text{Pr}\left\{\sum_{n=1}^N x_{k,n}r_{k,n}^{(t)} \ge
q_k\right\} \ge 1-\epsilon_k, \quad\forall k \label{chance}\\
&\quad\sum_{k=1}^K x_{k,n} \le 1, \quad\forall n \nonumber \\
&\quad x_{k,n}\ge0, \quad\forall k,n, \nonumber
\end{align}
where the expectation\footnote{In \eqref{obj1}, we replace the
time-average data rate $\bar{r}_{k,n}$ by its ensemble average
$\mathbb{E}\left\{r_{k,n}^{(t)}\right\}$ due to the ergodicity of
channel fading over the window.} in \eqref{obj1} is taken over
the random channel process $g=\{g_{k,n}^{(t)}\}$ for $t\in[t_0,t_0+T]$,
and $\epsilon_k\in[0,1]$ in \eqref{chance} is the maximum outage
probability user $k$ can tolerate. In the above formulation, we seek
the optimal SCA that maximizes the expected system throughout while
satisfying each user's short-term QoS requirement, i.e., the
instantaneous data rate of user $k$ is higher than $q_k$ with
probability at least $1-\epsilon_k$. The above formulation is a
\emph{chance constrained program} since a probabilistic constraint
\eqref{chance} has been imposed.

\section{Safe Tractable Constraints}\label{approx}
Despite its utility and relevance to real applications, the chance
constraint \eqref{chance} imposed in $\mathcal{P}_\text{slow}$ makes
the optimization highly intractable. The main reason is that the
convexity of the feasible set defined by \eqref{chance} is difficult
to verify. Indeed, given a generic chance constraint
$\text{Pr}\left\{F(\mathbf{x},\mathbf{r})>0\right\}\le \epsilon$
where $\mathbf{r}$ is a random vector, $\mathbf{x}$ is the vector of
decision variable, and $F$ is a real-valued function, its feasible
set is often non-convex except for very few special cases
\cite{BirLou:B97,NemSha:06}. Moreover, even with the nice function in
\eqref{chance}, i.e.,
$F(\mathbf{x},\mathbf{r})=q_k-\sum_{n=1}^Nx_{k,n}r_{k,n}^{(t)}$ is
bilinear in $\mathbf{x}$ and $\mathbf{r}$, with independent entries
$r_{k,n}^{(t)}$ in $\mathbf{r}$ whose distribution is known, it is
still unclear how to compute the probability in \eqref{chance}
efficiently.

To circumvent the above hurdles, we propose the following
formulation $\mathcal{\tilde{P}}_\text{slow}$ by replacing the
chance constraints \eqref{chance} with a system of constraints
$\mathcal{H}$ such that (i) $\mathbf{x}$ is feasible for
\eqref{chance} whenever it is feasible for $\mathcal{H}$, and (ii)
the constraints in $\mathcal{H}$ are convex and efficiently
computable\footnote{Condition (i) is referred to as ``safe''
condition, and condition (ii) is referred to as ``tractable''
condition.}. The new formulation is given as follows:
\begin{align}
\mathcal{\tilde{P}}_\text{slow}:\quad\max_{x_{k,n}} &
\quad\sum_{k=1}^K\sum_{n=1}^N
x_{k,n}\mathbb{E}\left\{r_{k,n}^{(t)}\right\} \label{obj2}\\
\text{s.t.}
&\quad\inf_{\varrho>0}\biggr\{q_k+\varrho\sum_{n=1}^N\Lambda_k(-\varrho^{-1}x_{k,n})\notag\\
&\hspace{1.5cm}-\varrho\log\epsilon_k\biggr\}\le0, \quad\forall k \label{chance.approx}\\
&\quad\sum_{k=1}^K x_{k,n} \le 1, \quad\forall n \label{xset21}\\
&\quad x_{k,n}\ge0, \quad\forall k,n, \label{xset22}
\end{align}
where $\Lambda_k(\cdot)$ is the cumulant generating function of
$r_{k,n}^{(t)}$,
\begin{align}\label{moment1}
\Lambda_k(-\varrho^{-1}\hat{x}_{k,n})=\log\biggr[\int_0^{\infty}&\left(1+\frac{p_t\xi}{\Gamma
N_0}\right)^{-\frac{W\hat{x}_{k,n}}{\varrho\ln2}}\notag\\
&\hspace{0.58cm}\cdot\frac{1}{\sigma_k}\exp\left(-\frac{\xi}
{\sigma_k}\right)d\xi\biggr].
\end{align}
In the following, we first prove that any solution $\mathbf{x}$ that is feasible for the STC \eqref{chance.approx}
in $\mathcal{\tilde{P}}_\text{slow}$ is also feasible for the chance
constraints \eqref{chance}. Then, we prove that
$\mathcal{\tilde{P}}_\text{slow}$ is convex.

\begin{proposition}\label{thm_apprx}
Suppose that $g_{k,n}^{(t)}$ (and hence $r^{(t)}_{k,n}$) are
independent random variables for different $n$ and $k$, where the
PDF of $g_{k,n}^{(t)}$ follows \eqref{expon}. Furthermore, given
$\epsilon_k>0$, suppose that there exists an
$\mathbf{\hat{x}}=[\hat{x}_{1,1},\cdots,\hat{x}_{N,1},
\ldots,\hat{x}_{1,K},\cdots,\hat{x}_{N,K}]^T\in\mathbb{R}^{NK}$ such
that
\begin{equation}\label{apx2}
G_k(\mathbf{\hat{x}})\!\triangleq\!\inf_{\varrho>0}\!
\left\{\!q_k\!+\!\varrho\sum_{n=1}^N\Lambda_k(-\varrho^{-1}\hat{x}_{k,n})
\!-\!\varrho\log\epsilon_k\!\right\}
\!\le\!0,\quad\forall k.
\end{equation}
Then, the allocation decision $\mathbf{\hat{x}}$ satisfies
\begin{equation}\label{chance1}
\Pr\left\{\sum_{n=1}^N \hat{x}_{k,n}r_{k,n}^{(t)} \ge q_k\right\}
\ge 1-\epsilon_k, \quad\forall k.
\end{equation}
\end{proposition}

\begin{proof}
Our argument will use the Bernstein approximation theorem proposed
in \cite{NemSha:06}.\footnote{For the reader's convenience, both the
theorem and a rough proof are provided in Appendix \ref{apdx_bern}.}
Suppose there exists an $\mathbf{\hat{x}}\in\mathbb{R}^{NK}$ such
that $G_k(\mathbf{\hat{x}})\le0$, i.e.,
\begin{eqnarray}
\inf_{\varrho>0}\left\{q_k+\varrho\sum_{n=1}^N\Lambda_k(-\varrho^{-1}\hat{x}_{k,n})
-\varrho\log\epsilon_k\right\}\le0. \label{eq:pf1}
\end{eqnarray}
The function inside the $\inf_{\varrho>0}\{\cdot\}$  is equal to
\begin{align}
&q_k+\varrho\sum_{n=1}^N\log\mathbb{E}\left\{\exp\biggr(-\varrho^{-1}
\hat{x}_{k,n}r_{k,n}^{(t)}\biggr)\right\}-\varrho\log\epsilon_k
\\
&\hspace{0.2cm}=
q_k+\varrho\log\mathbb{E}\left\{\exp\biggr(\varrho^{-1}\Big(-\sum_{n=1}^N
\hat{x}_{k,n}r_{k,n}^{(t)}\Big)\biggr)\right\}-\varrho\log\epsilon_k
\label{eq:indep}\\
&\hspace{0.2cm}=
\varrho\log\mathbb{E}\left\{\exp\biggr(\varrho^{-1}\Big(q_k-\sum_{n=1}^N
\hat{x}_{k,n}r_{k,n}^{(t)}\Big)\biggr)\right\}-\varrho\log\epsilon_k,
\qquad
\end{align}
where the expectation $\mathbb{E}\left\{\cdot\right\}$ can be
computed using the distributional information of $g_{k,n}^{(t)}$ in
\eqref{expon}, and \eqref{eq:indep} follows from the independence of
random variable $r_{k,n}^{(t)}$ over $n$.

Let $F_k(\mathbf{x},\mathbf{r})=q_k-\sum_{n=1}^N
x_{k,n}r_{k,n}^{(t)}$. Then, \eqref{eq:pf1} is equivalent to
\begin{eqnarray}
\inf_{\varrho>0}\left\{\varrho\mathbb{E}
\left\{\exp\left(\varrho^{-1}F_k(\mathbf{\hat{x}},\mathbf{r})
\right)\right\}-\varrho\epsilon_k\right\}\le0. \label{eq:pf2}
\end{eqnarray}
According to Theorem \ref{thm_bern} in Appendix \ref{apdx_bern}, the
chance constraints \eqref{chance1} hold if there exists a
$\varrho>0$ satisfying \eqref{eq:pf2}. Thus, the validity of
\eqref{chance1} is guaranteed by the validity of \eqref{apx2}.
\end{proof}

Now, we prove the convexity of \eqref{chance.approx} in the
following proposition.
\begin{proposition}\label{thm_cvx}
The constraints imposed in \eqref{chance.approx} are convex in
$\mathbf{x}=[x_{1,1},\cdots,x_{N,1},
\ldots,x_{1,K},\cdots,x_{N,K}]^T\in\mathbb{R}^{NK}$.
\end{proposition}

\begin{proof}
The convexity of \eqref{chance.approx} can be verified as
follows. We define the function inside the
$\inf_{\varrho>0}\{\cdot\}$ in \eqref{apx2} as
\begin{equation}
H_k(\mathbf{x},\varrho)\triangleq
q_k+\varrho\sum_{n=1}^N\Lambda_k(-\varrho^{-1}x_{k,n})-\varrho\log\epsilon_k,
\quad\forall k.
\end{equation}
It is easy to verify the convexity of $H_k(\mathbf{x},\varrho)$ in
$(\mathbf{x},\varrho)$, since the cumulant generating function is
convex. Hence, $G_k(\mathbf{x})$ in \eqref{apx2} is convex in
$\mathbf{x}$ due to the preservation of convexity by minimization
over $\varrho>0$.
\end{proof}

\section{Algorithm}\label{algo}
In this section, we propose an algorithm for solving Problem
$\mathcal{\tilde{P}}_\text{slow}$. In $\tilde{P}_\text{slow}$, the
STC \eqref{chance.approx} arises as a subproblem, which by itself
requires a minimization over $\varrho$. Hence, despite its
convexity, the entire problem $\tilde{P}_\text{slow}$ cannot be
trivially solved using standard solvers of convex optimization.
This is due to the fact that the subproblem introduces difficulties,
for example, in defining the barrier function in
\emph{path-following algorithms} or providing the (sub-)gradient in
\emph{primal-dual methods} (see \cite{ResPar:B06} for details of
these algorithms). Fortunately, we can employ \emph{interior point
cutting plane methods} to solve Problem
$\mathcal{\tilde{P}}_\text{slow}$ (see \cite{Mit:03} for a
survey).  Before we delve into the details, let us briefly sketch
the principles of the algorithm as follows.

\begin{algorithm}[t]
	\vspace{0.1cm}
\caption{Structure of the Proposed Algorithm} \label{alg.sketch}
\begin{algorithmic}[1]
 \REQUIRE The feasible solution set of Problem
 $\mathcal{\tilde{P}}_\text{slow}$ is a compact set $\mathcal{X}$ defined by
 \eqref{chance.approx}-\eqref{xset22}.
 \STATE Construct a polytope $X^0\supset\mathcal{X}$ by \eqref{xset21}-\eqref{xset22}. Set $i\leftarrow 0$.
 \STATE Choose a query point (\emph{Subsection \ref{algo}\!.\,A-1}) at the $i$th iteration as $\mathbf{x}^i$
 by computing the analytic center of $X^i$. Initially, set $\mathbf{x}^0=\mathbf{e}/K\in X^0$ where $\mathbf{e}$ is an $N$-vector of ones.
 \STATE Query the separation oracle (\emph{Subsection \ref{algo}\!.\,A-2}) with $\mathbf{x}^i$:
 \IF{$\mathbf{x}^i\in\mathcal{X}$}
    \STATE generate a hyperplane (optimality cut) through $\mathbf{x}^i$ to remove the part of $X^i$ that has lower objective
    values
 \ELSE
    \STATE generate a hyperplane (feasibility cut) through $\mathbf{x}^i$ to remove the part of $X^i$ that contains infeasible
    solutions.
 \ENDIF
 \STATE Set $i\leftarrow i+1$, and update $X^{i+1}$ by the separation hyperplane.
 \IF{termination criterion (\emph{Subsection \ref{algo}\!.\,B}) is satisfied}
    \STATE stop
 \ELSE
    \STATE return to step 2.
 \ENDIF
\end{algorithmic}
\vspace{0.15cm}
\end{algorithm}

Suppose that we would like to find a point $\mathbf{x}$ that is
feasible for \eqref{chance.approx}-\eqref{xset22} and is within a
distance of $\delta>0$ to an optimal solution $\mathbf{x}^*$ of
$\mathcal{\tilde{P}}_\text{slow}$, where $\delta>0$ is an error
tolerance parameter (i.e., $\mathbf{x}$ satisfies
$||\mathbf{x}-\mathbf{x}^*||_2<\delta$). We maintain the invariant
that at the beginning of each iteration, the feasible set is
contained in some polytope (i.e., a bounded polyhedron). Then, we
generate a query point inside the polytope and ask a ``separation
oracle'' whether the query point belongs to the feasible set. If
not, then the separation oracle will generate a so-called separating
hyperplane through the query point to cut out the polytope, so that
the remaining polytope contains the feasible set.\footnote{Note that
such a separating hyperplane exists due to the convexity of the
feasible set \cite{HirLem:B01}.}  Otherwise, the separation oracle
will return a hyperplane through the query point to cut out the
polytope towards the opposite direction of improving objective
values.

We can then proceed to the next iteration with the new polytope.
To keep track of the progress, we can use the so-called potential
value of the polytope. Roughly speaking, when the potential value
becomes large, the polytope containing the feasible set has become
small. Thus, if the potential value exceeds a certain threshold, so
that the polytope is negligibly small, then we can terminate the
algorithm. As will be shown later, such an algorithm will in fact
terminate in a polynomial number of steps.

We now give the structure of the algorithm. A detailed flow chart is
shown in Fig. \ref{fg_algo} for readers' interest.

\begin{figure}[h]
\begin{center}
	\psfrag{x1}[][]{\footnotesize{$X^0\!:\!(A^0,\mathbf{b}^0)$}}
 	\psfrag{x2}[][]{\footnotesize{$\mathbf{x}^0=\mathbf{e}/K.$}}
	\psfrag{x3}[][]{\footnotesize{$\varrho^*\!=\!\arg\inf\limits_{\varrho>0}[H(\mathbf{x}^i,\varrho)]$}}
	\psfrag{x4}[][]{\footnotesize{$H(\mathbf{x}^i,\varrho^*)\!\leq\!0$}}
	\psfrag{x5}[][]{\footnotesize{$X^{i+1}\!\!:\!(\!A^{i+1}\!,\!\mathbf{b}^{i+1}\!)$}}
 	\psfrag{x6}[][]{\footnotesize{$X^{i+1}.$}}
	\psfrag{x7}[][]{\footnotesize{$\mathbf{x}^{i+1},A^{i+1}\!,\mathbf{b}^{i+1}$}}
 	\psfrag{x8}[][]{\footnotesize{$\mathbf{x}^{i+1}$}}
 	\psfrag{x9}[][]{\footnotesize{$(A^{i+1},\mathbf{b}^{i+1}).$}}
 	\psfrag{x10}[][]{\footnotesize{$\mathbf{x}^*\!=\!\mathbf{x}^i.$}}
\includegraphics [height=15cm]{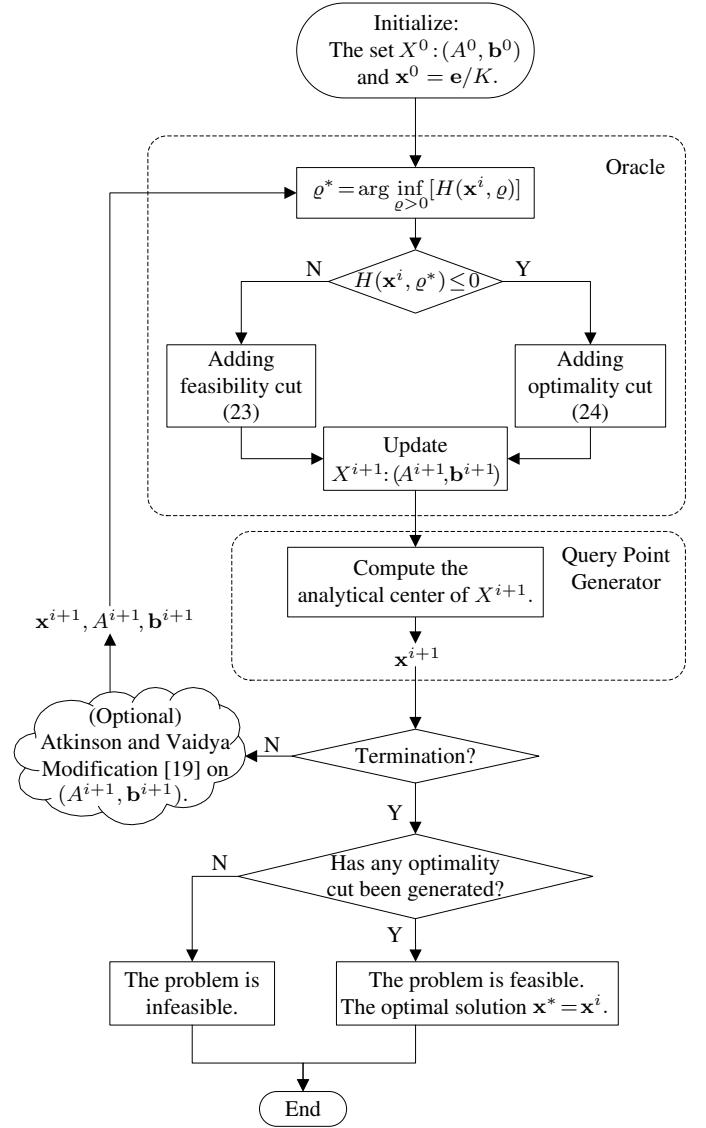}
\caption{Flow chart of the algorithm for solving Problem $\mathcal{\tilde{P}}_\text{slow}$.} \label{fg_algo}
\end{center}
\end{figure}

\subsection{The Cutting-Plane-Based Algorithm}
\emph{1) Query Point Generator: (Step 2 in Algorithm
\ref{alg.sketch})}

In each iteration, we need to generate a query point inside the
polytope $X^i$. For algorithmic efficiency, we adopt the analytic
center (AC) of the containing polytope as the query point
\cite{GonMerSarVia:96}. The AC of the polytope
$X^i=\{\mathbf{x}\in\mathbb{R}^{NK}: A^i\mathbf{x}\le
\mathbf{b}^i\}$ at the $i$th iteration is the unique solution
$\mathbf{x}^i$ to the following convex problem:
\begin{align}\label{ac}
\max_{\{\mathbf{x}^i,\mathbf{s}^i\}} & \quad\sum_{m=1}^{M^{i}}\log s_m^i \\
\text{s.t.} & \quad\mathbf{s}^i=\mathbf{b}^i-A^i\mathbf{x}^i. \notag
\end{align}
We define the optimal value of the above problem as the potential
value of the polytope $X^i$. Note that the uniqueness of the
analytic center is guaranteed by the strong convexity of the
potential function
$\mathbf{s}^i\mapsto-\sum_{m=1}^{M^{i}}\log s_m^i$, assuming that
$X^i$ is bounded and has a non-empty interior. The AC of a polytope
can be viewed as an approximation to the geometric center of the
polytope, and thus any hyperplane through the AC will separate the
polytope into two parts with roughly the same volume.

Although it is computationally involved to directly solve \eqref{ac}
in each iteration, it is shown in \cite{GofLuoYe:96} that an
approximate AC is sufficient for our purposes, and that an
approximate AC for the $(i+1)$st iteration can be obtained from an
approximate AC for the $i$th iteration by applying $\mathcal{O}(1)$
Newton steps.

\emph{2) The Separation Oracle: (Steps 3-8 in Algorithm
\ref{alg.sketch})}

The oracle is a major component of the algorithm that plays two
roles: checking the feasibility of the query point, and generating
cutting planes to cut the current set.

\begin{itemize}
\item \emph{Feasibility Check}
\end{itemize}

We write the constraints of $\mathcal{\tilde{P}}_\text{slow}$ in a
condensed form as follows:
\begin{align}
&G_k(\mathbf{x})=\inf_{\varrho>0}\left\{H_k(\mathbf{x},\varrho)\right\}\le0, \quad\forall k \label{fb1}\\
&A^0~\mathbf{x}\le \mathbf{b}^0 \label{fb2}
\end{align}
where
\begin{align*}
A^0&=\begin{bmatrix}I_N~I_N\cdots~
	I_N\\-I_{NK}\end{bmatrix}\in\mathbb{R}^{(N+NK)\times NK},\\
\mathbf{b}^0&=[\mathbf{e}_N^T,
	~\mathbf{0}_{NK}^T]^T\in\mathbb{R}^{N+NK}
\end{align*}
with $I_N$ and $\mathbf{e}_N$ denoting the $N\times N$ identity matrix and $N$-vector of ones respectively, and \eqref{fb2} is the combination\footnote{To reduce numerical
errors in computation, we suggest normalizing each constraint in
\eqref{fb2}.} of \eqref{xset21} and \eqref{xset22}. Now, we first use \eqref{fb2} to construct a
relaxed feasible set via
\begin{equation}
X^0=\{\mathbf{x}\in\mathbb{R}^{NK}: A^0~\mathbf{x}\le
\mathbf{b}^0\}.
\end{equation}
Given a query point $\mathbf{x}\in X^0$, we can verify its
feasibility to $\mathcal{\tilde{P}}_\text{slow}$ by checking if it
satisfies \eqref{fb1}, i.e., if
$\inf_{\varrho>0}\{H_k(\mathbf{x},\varrho)\}$ is no larger than 0.
This requires solving a minimization problem over $\varrho>0$. Due
to the unimodality of $H_k(\mathbf{x},\varrho)$ in $\varrho$, we can
simply take a line search procedure, e.g., using Golden-section
search or Fibonacci search, to find the minimizer $\varrho^*$. The
line search is more efficient when compared with
derivative-based algorithms, since only function
evaluations\footnote{The cumulant generating function
$\Lambda_k(\cdot)$ in \eqref{moment1} can be evaluated numerically,
e.g., using rectangular rule, trapezoid rule, or Simpson's rule,
etc.} are needed during the search.

\begin{itemize}
\item \emph{Cutting Plane Generation}
\end{itemize}

In each iteration, we generate a cutting plane, i.e., a hyperplane
through the query point, and add it as an additional constraint to
the current polytope $X^i$. By adding cutting plane(s) in each
iteration, the size of the polytope keeps shrinking. There are two
types of cutting planes in the algorithm depending on the
feasibility of the query point.

If the query point $\mathbf{x}^i\in X^i$ is infeasible, then a
hyperplane called \emph{feasibility cut} is generated at
$\mathbf{x}^i$ as follows:
\begin{equation}\label{fbcut}
\biggr(\frac{\mathbf{u}^{i,\bar{\kappa}}}{||\mathbf{u}^{i,\bar{\kappa}}||}\biggr)^T(\mathbf{x}-\mathbf{x}^i)\le0,
\quad\forall \bar{\kappa}\in \bar{K},
\end{equation}
where $||\cdot||$ is the Euclidean norm,
$\bar{K}=\{k:H_k(\mathbf{x}^i,t^*)>0,~k=1,2,\cdots,K\}$ is the set
of users whose chance constraints are violated, and
$\mathbf{u}^{i,\bar{\kappa}}=[u_{1,k}^{i,\bar{\kappa}},\cdots,u_{N,1}^{i,\bar{\kappa}},\ldots,u_{1,K}^{i,\bar{\kappa}},\cdots,u_{N,K}^{i,\bar{\kappa}}]^T
\in\mathbb{R}^{NK}$ is the gradient of
$G_{\bar{\kappa}}(\mathbf{x})$ with respect to $\mathbf{x}$, i.e.,
\begin{align}
&\!u_{k,n}^{i,\bar{\kappa}}=\frac{\partial H_{\bar{\kappa}}(\mathbf{x},\varrho^*)}{\partial x_{k,n}}\biggr|_{x_{k,n}=x_{k,n}^i} \notag\\
&\!=\frac{-\frac{W}{\ln2}\int_0^{\infty}\!\biggr(\!1\!+\!\frac{p_t\xi}{\Gamma
N_0}\!\biggr)^{-\frac{Wx_{k,n}^i}{\varrho^*\ln2}}
\!\!\ln\!\biggr(\!1\!+\!\frac{p_t\xi}{\Gamma N_0}\!\biggr)
\!\frac{1}{\sigma_{\bar{\kappa}}}\exp\!\left(-\frac{\xi}
{\sigma_{\bar{\kappa}}}\right)\!d\xi}{\int_0^{\infty}\biggr(1+\frac{p_t\xi}{\Gamma
N_0}\biggr)^{-\frac{Wx_{k,n}^i}{\varrho^*\ln2}}\frac{1}{\sigma_k}\exp\left(-\frac{\xi}
{\sigma_{\bar{\kappa}}}\right)d\xi}. \nonumber
\end{align}
The reason we call \eqref{fbcut} a feasibility cut(s) is that any
$\mathbf{x}$ which does not satisfy \eqref{fbcut} must be infeasible
and can hence be dropped.

If the point $\mathbf{x}^i$ is feasible, then an \emph{optimality
cut} is generated as follows:
\begin{equation}\label{optcut}
\quad\biggr(\frac{\mathbf{v}}{||\mathbf{v}||}\biggr)^T(\mathbf{x}-\mathbf{x}^i)\le0,
\end{equation}
where
$\mathbf{v}=\big[-\mathbb{E}\{r_{1,1}^{(t)}\},\cdots,-\mathbb{E}\{r_{N,1}^{(t)}\},\ldots,
-\mathbb{E}\{r_{1,K}^{(t)}\},\cdots,$
$-\mathbb{E}\{r_{N,K}^{(t)}\}\big]^T\in\mathbb{R}^{NK}$
is the derivative of the objective of
$\mathcal{\tilde{P}}_\text{slow}$ in \eqref{obj2} with respect to
$\mathbf{x}$. The reason we call \eqref{optcut} an optimality cut is
that any optimal solution $\mathbf{x}^*$ must satisfy \eqref{optcut}
and hence any $\mathbf{x}$ which does not satisfy \eqref{optcut} can
be dropped.

Once a cutting plane is generated according to \eqref{fbcut} or
\eqref{optcut}, we use it to update the polytope $X^i$ at the $i$th
iteration as follows
$$
X^i=\{\mathbf{x}\in\mathbb{R}^{NK}: A^i\mathbf{x}\le\mathbf{b}^i\}.
$$
Here, $A^i$ and $\mathbf{b}^i$ are obtained by adding the cutting
plane to the previous polytope $X^{i-1}$. Specifically, if the
oracle provides a feasibility cut as in \eqref{fbcut}, then
\begin{align*}
A^i&=\begin{bmatrix}A^{i-1}\\(\mathbf{u}_k^i/||\mathbf{u}_k^i||)^T\end{bmatrix}
\in\mathbb{R}^{(M^{i-1}+|\bar{K}|)\times NK},\quad\\
b^i&=\begin{bmatrix}b^{i-1}\\(\mathbf{u}_k^i/||\mathbf{u}_k^i||)^T\mathbf{x}^i\end{bmatrix}
\in\mathbb{R}^{M^{i-1}+|\bar{K}|}
\end{align*}
where $M_{i-1}$ is the number of rows in $A_{i-1}$, and $|\cdot|$ is
the number of elements contained in the given set; if the oracle
provides an optimality cut as in \eqref{optcut}, then
\begin{align*}
A^i&=\begin{bmatrix}A^{i-1}\\(\mathbf{v}/||\mathbf{v}||)^T\end{bmatrix}
\in\mathbb{R}^{(M^{i-1}+1)\times NK},\quad\\
b^i&=\begin{bmatrix}b^{i-1}\\(\mathbf{v}/||\mathbf{v}||)^T\mathbf{x}^i\end{bmatrix}
\in\mathbb{R}^{M^{i-1}+1}.
\end{align*}

\subsection{Global Convergence \!\&\! Complexity (Step \!10 in Algorithm \!\ref{alg.sketch})}
In the following, we investigate the convergence properties of the
proposed algorithm. As mentioned earlier, when the polytope is too
small to contain a full-dimensional closed ball of radius
$\delta>0$, the potential value will exceeds a certain threshold.
Then, the algorithm can terminate since the query point is within a
distance of $\delta>0$ to some optimal solution of
$\tilde{\mathcal{P}}_{slow}$. Such an idea is formalized in
\cite{GofLuoYe:96}, where it was shown that the analytic center-based
cutting plane method can be used to solve convex programming
problems in polynomial time.  Upon following the proof in
\cite{GofLuoYe:96}, we obtain the following result:
\begin{theorem}\label{term}
(cf. \cite{GofLuoYe:96}) Let $\delta>0$ be the error tolerance
parameter, and let $m$ be the number of variables.  Then,
Algorithm \ref{alg.sketch} terminates with a solution $\mathbf{x}$
 that is feasible for $\mathcal{\tilde{P}}_\text{slow}$ and satisfies
$\|\mathbf{x}-\mathbf{x}^*\|_2<\delta$ for some optimal solution
$\mathbf{x}^*$ to $\mathcal{\tilde{P}}_\text{slow}$ after at most
$\mathcal{O}((m/\delta)^2)$ iterations.
\end{theorem}

Thus, the proposed algorithm can solve Problem
$\mathcal{\tilde{P}}_\text{slow}$ within
$\mathcal{O}((NK/\delta)^2)$ iterations. It turns out that the
algorithm can be made considerably more efficient by dropping
constraints that are deemed ``unimportant'' in \cite{AtkVai:95}. By
incorporating such a strategy in Algorithm \ref{alg.sketch}, the
total number of iterations needed by the algorithm can be reduced to
$\mathcal{O}(NK\log^2(1/\delta))$. We refer the readers to
\cite{AtkVai:95,Mit:03} for details.

\subsection{Complexity Comparison between Slow and Fast Adaptive OFDMA}
It is interesting to compare the complexity of slow and fast
adaptive OFDMA schemes formulated in
$\mathcal{\tilde{P}_\text{slow}}$ and $\mathcal{P}_\text{fast}$,
respectively. To obtain an optimal solution to
$\mathcal{P}_\text{fast}$, we need to solve a linear program (LP).
This requires $\mathcal{O}(\sqrt{NK}L_0)$ iterations, where $L_0$ is
number of bits to store the data defining the LP \cite{Ye:B97}. At
first glance, the iteration complexity of solving a fast adaptation
$\mathcal{P}_\text{fast}$ can be lower than that of solving
$\mathcal{\tilde{P}}_\text{slow}$ when the number of users or
subcarriers are large. However, it should be noted that only one
$\mathcal{\tilde{P}}_\text{slow}$ needs to be solved for each
adaptation window, while $\mathcal{P}_\text{fast}$ has to be solved
for each time slot. Since the length of adaptation window is equal
to $T$ time slots, the overall complexity of the slow adaptive OFDMA
can be much lower than that of conventional fast adaptation schemes,
especially when $T$ is large.

Before leaving this section, we emphasize that the advantage of slow
adaptive OFDMA lies not only in computational cost reduction, but
also in reducing control signaling overhead. We will investigate this
in more detail in Section \ref{simu}.

\section{Problem Size Reduction}\label{reduct}
In this section, we show that the problem size of
$\mathcal{\tilde{P}}_\text{slow}$ can be reduced from
$NK$ variables to $K$ variables under some mild assumptions. Consequently,
the computational complexity of slow adaptive OFDMA can be
markedly lower than that of fast adaptive OFDMA.

In practical multicarrier systems, the frequency intervals between
any two subcarriers are much smaller than the carrier frequency. The
reflection, refraction and diffusion of electromagnetic waves behave
the same across the subcarriers. This implies that the channel
gain $g_{k,n}^{(t)}$ is identically distributed over $n$
(subcarriers), although it is not needed in our algorithm
derivations in the previous sections.

When $g_{k,n}^{(t)}$ for different $n$ are identically distributed,
different subcarriers become indistinguishable to a user $k$. In
this case, the optimal solution, if exists, does not depend on $n$.
Replacing $x_{k,n}$ by $x_k$ in $\mathcal{\tilde{P}}_\text{slow}$,
we obtain the following formulation:
\begin{align*}
\mathcal{\tilde{P}}'_\text{slow}\!:\,\max_{x_k} &~
\sum_{k=1}^K\sum_{n=1}^N x_k\mathbb{E}\left\{r_{k,n}^{(t)}\right\}\\
\text{s.t.} &~ \inf_{\varrho>0}\!\left\{q_k\!+\!\varrho N\Lambda_k(-\varrho^{-1}x_k)\!-\!\varrho\log\epsilon_k\right\}\!\le\!0, \quad\forall k \\
&~ \sum_{k=1}^K x_k \le 1, \\
&~ x_k\ge0, \quad\forall k.
\end{align*}

Note that the problem structure of
$\mathcal{\tilde{P}}'_\text{slow}$ is exactly the same as that of
$\mathcal{\tilde{P}}_\text{slow}$, except that the problem size is
reduced from $NK$ variables to $K$ variables. Hence, the algorithm
developed in Section \ref{algo} can also be applied to solve
$\mathcal{\tilde{P}}'_\text{slow}$, with the following vector/matrix
size reductions:
$A^0=[\mathbf{e}_N,-I_K]^T\in\mathbb{R}^{(1+K)\times K}$,
$\mathbf{b}^0=[1,~0,\cdots,~0]^T\in\mathbb{R}^{1+K}$ in \eqref{fb2},
$\mathbf{u}^{i,\bar{\kappa}}=[u_1^{i,\bar{\kappa}},\cdots,u_K^{i,\bar{\kappa}}]^T\in\mathbb{R}^K$
in \eqref{fbcut}, and
$\mathbf{v}=\big[-\mathbb{E}\{r_1^{(t)}\},\cdots,-\mathbb{E}\{r_K^{(t)}\}\big]^T\in\mathbb{R}^K$
in \eqref{optcut}. Compared with $\mathcal{\tilde{P}}_\text{slow}$,
the iteration complexity of $\mathcal{\tilde{P}}'_\text{slow}$ is
now reduced to $\mathcal{O}(K\log^2(1/\delta))$. Indeed, this can
even be lower than the complexity of solving one
$\mathcal{P}_\text{fast}$
--- $\mathcal{O}(\sqrt{NK}L_0)$, since $K$ is typically much
smaller than $N$ in real systems. Thus, the overall complexity of
slow adaptive OFDMA is significantly lower than that of fast adaptation
over $T$ time slots.

Before leaving this section, we emphasize that the problem size
reduction in $\mathcal{\tilde{P}}'_\text{slow}$ does not compromise
the optimality of the solution. On the other hand,
$\mathcal{\tilde{P}}_\text{slow}$ is more general in the sense that
it can be applied to systems in which the frequency bands of parallel
subchannels are far apart, so that the channel distributions are
not identical across different subchannels.

\section{Simulation Results}\label{simu}
In this section, we demonstrate the performance of our proposed slow
adaptive OFDMA scheme through numerical simulations. We simulate an
OFDMA system with $4$ users and $64$ subcarriers. Each user $k$ has
a requirement on its short-term data rate $q_k=20\text{bps}$. The
$4$ users are assumed to be uniformly distributed in a cell of
radius $R=100\text{m}$. That is, the distance $d_k$ between user $k$
and the BS follows the distribution\footnote{The distribution of
user's distance from the BS $f(d)=\frac{2d}{R^2}$ is derived from
the uniform distribution of user's position $f(x,y)=\frac{1}{\pi
R^2}$, where $(x,y)$ is the Cartesian coordinate of the position.}
$f(d)=\frac{2d}{R^2}$. The path-loss exponent $\gamma$ is equal to
4, and the shadowing effect $s_k$ follows a log-normal distribution,
i.e., $10\log_{10}(s_k)\sim\mathcal{N}(0,8\text{dB})$. The
small-scale channel fading is assumed to be Rayleigh distributed.
Suppose that the transmission power of the BS on each subcarrier is
90dB measured at a reference point 1 meter away from the BS, which
leads to an average received power of 10dB at the boundary of the
cell\footnote{The average received power at the boundary is
calculated by
$90\text{dB}+10\log_{10}\left(\frac{100}{1}\right)^{-4}\text{dB}=10\text{dB}$
due to the path-loss effect.}. In addition, we set $W=1\text{Hz}$
and $N_0=1$, and the capacity gap is
$\Gamma=-\log(5\text{BER}) / 1.5=5.0673$, where the target BER
is set to be $10^{-4}$. Moreover, the length of one \emph{slot},
within which the channel gain remains unchanged, is
$T_0=1\text{ms}$.\footnote{The coherence time is
given by $T_0=\frac{9c}{16\pi f_c v}$, where $c$ is the speed
of light, $f_c$ is the carrier frequency, and $v$ is the velocity of
mobile user. As an example, we choose $f_c=2.5\text{GHz}$, and if the
user is moving at 45 miles per hour, the coherence time is around
1ms.} The length of the \emph{adaptation window} is
chosen to be $T=1\text{s}$, implying that each window contains
$1000$ slots. Suppose that the path loss and shadowing do not change
within a window, but varies independently from one window to
another. For each window, we solve the size-reduced problem
$\mathcal{\tilde{P}}'_\text{slow}$, and later Monte-Carlo simulation
is conducted over 61 independent windows that yield non-empty
feasible sets of $\mathcal{\tilde{P}}'_\text{slow}$ when
$\epsilon_k=0.1$.

\begin{figure}
\begin{center}
\includegraphics [height=7cm]{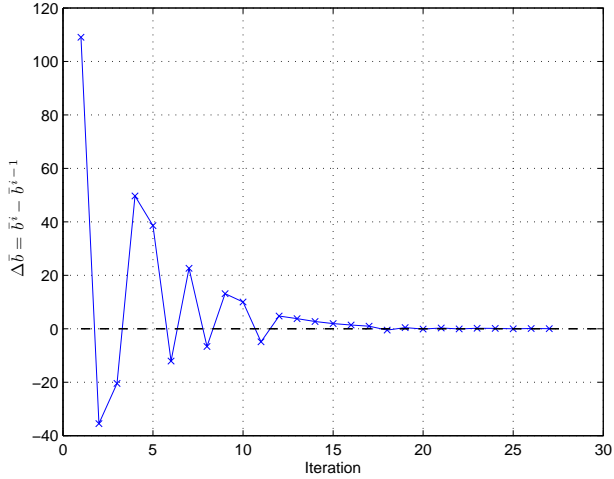}
\caption{Trace of the difference of objective value $\bar{b}^i$
between adjacent iterations ($\epsilon_k=0.2$).} \label{fg_cvgb}
\end{center}
\end{figure}

In Fig. \ref{fg_cvgb} and Fig. \ref{fg_cvgiter}, we investigate the
fast convergence of the proposed algorithm. The error tolerance
parameter is chosen as $\delta=10^{-2}$. In Fig. \ref{fg_cvgb}, we
record the trace of one adaptation window\footnote{The simulation
results show that all the feasible windows appear with similar
convergence behavior.} and plot the improvement in the objective
function value (i.e., system throughput) in each iteration, i.e.,
$\Delta\bar{b}=\bar{b}^i-\bar{b}^{i-1}$. When $\Delta\bar{b}$ is
positive, the objective value increases with each iteration. It can
be seen that $\Delta\bar{b}$ quickly converges to close to zero
within only 27 iterations. We also notice that fluctuation exists in
$\Delta\bar{b}$ within the first 11 iterations. This is mainly
because during the search for an optimal solution, it is possible
for query points to become infeasible. However, the feasibility cuts
\eqref{fbcut} then adopted will make sure that the query points in
subsequent iterations will eventually become feasible. The curve in
Fig. \ref{fg_cvgb} verifies the tendency. As
$\mathcal{\tilde{P}}_\text{slow}$ is convex, this observation
implies that the proposed algorithm can converge to an optimal
solution of $\mathcal{\tilde{P}}_\text{slow}$ within a small number
of iterations. In Fig. \ref{fg_cvgiter}, we plot the number of
iterations needed for convergence for different application windows.
The result shows that the proposed algorithm can in general converge
to an optimal solution of $\mathcal{\tilde{P}}_\text{slow}$ within
35 iterations. On average, the algorithm converges after 22
iterations, where each iteration takes 1.467 seconds.\footnote{We
conduct a simulation on Matlab 7.0.1, where the system
configurations are given as: Processor: Intel(R) Core(TM)2 CPU
P8400@2.26GHz 2.27GHz, Memory: 2.00GB, System Type: 32-bit Operating
System.}

\begin{figure}
\begin{center}
\includegraphics [height=7cm]{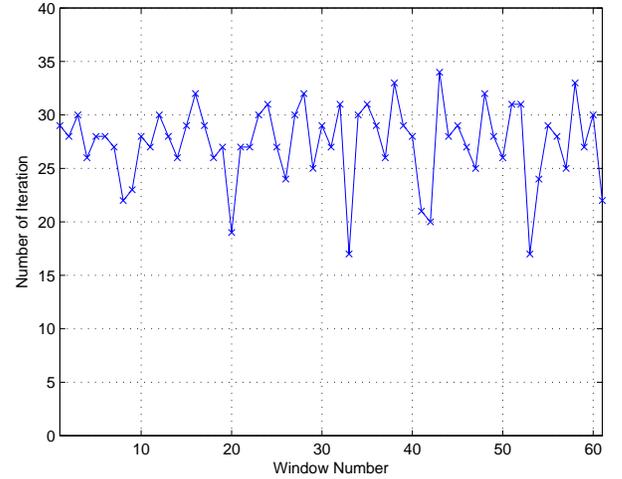}
\caption{Number of iterations for convergence of all the feasible
windows ($\epsilon_k=0.2$).} \label{fg_cvgiter}
\end{center}
\end{figure}

Moreover, we plot the number of iterations needed for checking the
feasibility of $\mathcal{\tilde{P}}_\text{slow}$. In Fig.
\ref{fg_cvgfb}, we conduct a simulation over 100 windows, which
consists of 61 feasible windows (dots with cross) and 39 infeasible
windows (dots with circle). On average, the algorithm can determine
if $\mathcal{\tilde{P}}_\text{slow}$ is feasible or not after 7
iterations. The quick feasibility check can help to deal with the
admission of mobile users in the cell. Particularly, if there is a
new user moving into the cell, the BS can adopt the feasibility
check to quickly determine if the radio resources can accommodate
the new user without sacrificing the current users' QoS
requirements.

\begin{figure}
\begin{center}
\includegraphics [height=7cm]{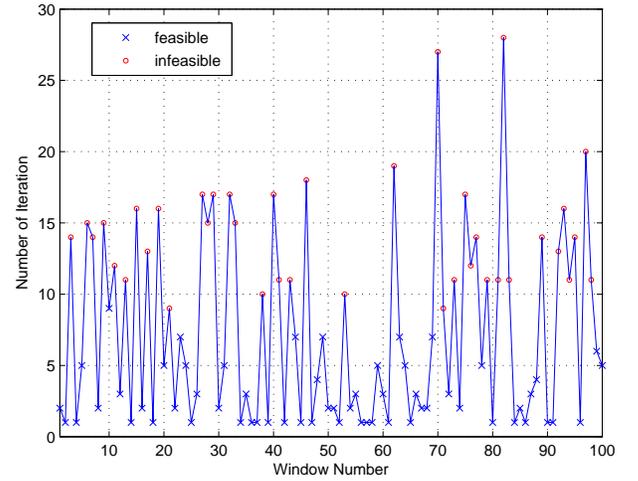}
\caption{Number of iterations for feasibility check of all the
windows ($\epsilon_k=0.2$).} \label{fg_cvgfb}
\end{center}
\end{figure}

In Fig. \ref{fg_thp.comp}, we compare the spectral efficiency of
slow adaptive OFDMA with that of fast adaptive OFDMA\footnote{For
illustrative purpose, we have only considered
$\mathcal{P}_\text{fast}$ as one of the typical formulations of fast
adaptive OFDMA in our comparisons. However, we should point out that
there are some work on fast adaptive OFDMA which impose less
restrictive constraints on user data rate requirement. For example,
in \cite{WonEva:08}, it considered average user data rate constraints
which exploits time diversity to achieve higher spectral
efficiency.}, where zero outage of short-term data rate requirement
is ensured for each user. In addition, we take into account the
control overheads for subcarrier allocation, which will considerably
affect the system throughput as well. Here, we assume that the
control signaling overhead consumes a bandwidth equivalent to $10\%$
of a slot length $T_0$ every time SCA is updated \cite{GroGeeKarWol:06}.
Note that within each window that contains $1000$ slots, the control
signaling has to be transmitted $1000$ times in the fast adaptation
scheme, but once in the slow adaptation scheme. In Fig.
\ref{fg_thp.comp}, the line with circles represents the performance
of the fast adaptive OFDMA scheme, while that with dots corresponds
to the slow adaptive OFDMA. The figure shows that although slow
adaptive OFDMA updates subcarrier allocation 1000 times less
frequently than fast adaptive OFDMA, it can achieve on average
71.88\% of the spectral efficiency. Considering the substantially
lower computational complexity and signaling overhead, slow adaptive
OFDMA holds significant promise for deployment in real-world
systems.

\begin{figure}
\begin{center}
\includegraphics [height=7cm]{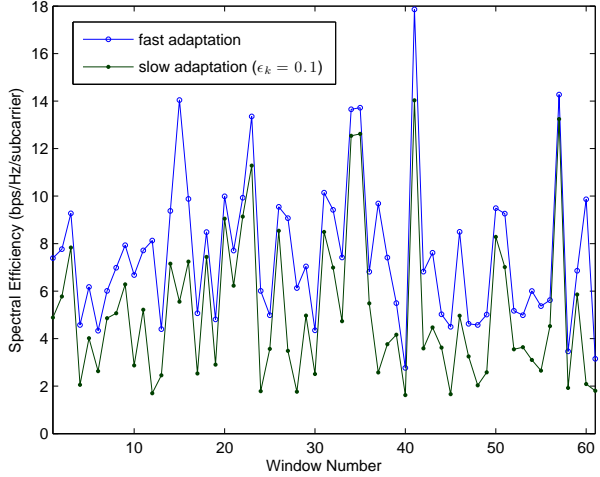}
\caption{Comparison of system spectral efficiency between fast
adaptive OFDMA and slow adaptive OFDMA.} \label{fg_thp.comp}
\end{center}
\end{figure}

As mentioned earlier, $\mathcal{\tilde{P}}_\text{slow}$ is more
conservative than the original problem $\mathcal{P}_\text{slow}$,
implying that the outage probability is guaranteed to be satisfied
if subcarriers are allocated according to the optimal solution of
$\mathcal{\tilde{P}}_\text{slow}$. This is illustrated in Fig.
\ref{fg_outage_eps}, which shows that the outage probability is
always lower than the desired threshold $\epsilon_k=0.1$.

Fig. \ref{fg_outage_eps} shows that the subcarrier allocation via $\mathcal{\tilde{P}}_\text{slow}$ could still be quite conservative, as the actual outage probability is much lower than
$\epsilon_k$. One way to tackle the problem is to set $\epsilon_k$ to
be larger than the actual desired value. For example, we could tune
$\epsilon_k$ from 0.1 to 0.3. By doing so, one can potentially
increase the system spectral efficiency, as the feasible set of
$\mathcal{\tilde{P}}_\text{slow}$ is enlarged. A question that
immediately arises is how to choose the right $\epsilon_k$, so that
the actual outage probability stays right below the desired value.
Towards that end, we can perform a binary search on $\epsilon_k$ to
find the best parameter that satisfies the requirement. Such a
search, however, inevitably involves high computational costs. On
the other hand, Fig. \ref{fg_thp_eps} shows that the gain in
spectral efficiency by increasing $\epsilon_k$ is marginal. The gain
is as little as 0.5 bps/Hz/subcarrier when $\epsilon_k$ is increased
drastically from 0.05 to 0.7. Hence, in practice, we can simply set
$\epsilon_k$ to the desired outage probability value to guarantee
the QoS requirement of users.

\begin{figure}
\begin{center}
\includegraphics [height=7cm]{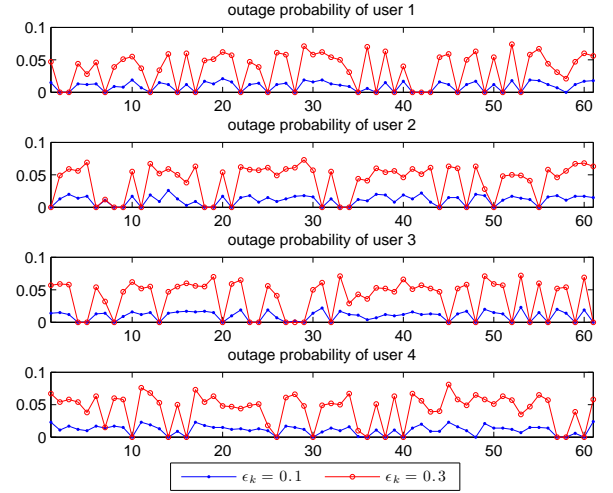}
\caption{Outage probability of the 4 users over 61 independent
feasible windows.} \label{fg_outage_eps}
\end{center}
\end{figure}

\begin{figure}
\begin{center}
\includegraphics [height=7cm]{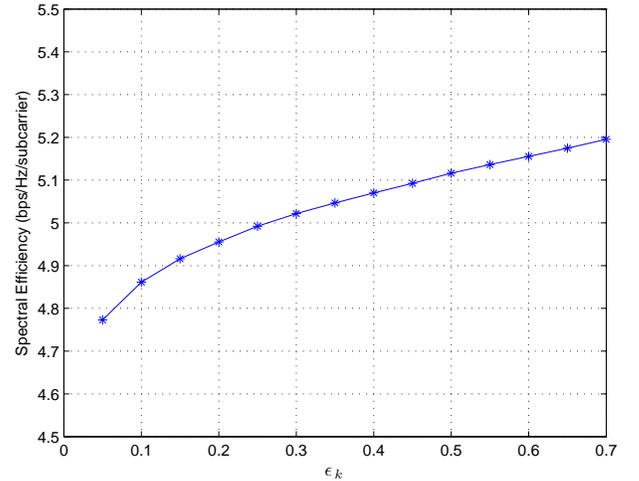}
\caption{Spectral efficiency versus tolerance parameter
$\epsilon_k$. Calculated from the average overall system throughput
on one window, where the long-term average channel gain $\sigma_k$
of the 4 users are $-65.11$dB, $-56.28$dB, $-68.14$dB and
$-81.96$dB, respectively.} \label{fg_thp_eps}
\end{center}
\end{figure}

In the development of the STC \eqref{chance.approx}, we considered
that the channel gain $g_{k,n}$ are independent for different $n$'s
and $k$'s. While it is true that channel fading is independent
across different users, it is typically correlated in the frequency
domain. We investigate the effect of channel correlation in
frequency domain through simulations. A wireless channel with an
exponential decaying power profile is adopted, where the
root-mean-square delay is equal to 37.79ns. For comparison, the
curves of outage probability with and without frequency correlation
are both plotted in Fig. \ref{fg_freq.corr}. We choose the tolerance
parameter to be $\epsilon_k=0.3$. The figure shows that with
frequency-domain correlation, the outage probability requirement of 0.3 is
violated occasionally. Intuitively, such a problem becomes
negligible when the channel is highly frequency selective, and is
more severe when the channel is more frequency flat. To address the
problem, we can set $\epsilon_k$ to be lower than the desired outage
probability value\footnote{Alternatively, we can divide $N$
subcarriers into $\frac{N}{N_c}$ subchannels (each subchannel
consists $N_c$ subcarriers), and represent each subchannel via an
average gain. By doing so, we can treat the subchannel gains as being
independent of each other.}.  For example, when we choose
$\epsilon_k=0.1$ in Fig. \ref{fg_freq.corr}, the outage
probabilities all decreased to lower than the desired value 0.3, and
hence the QoS requirement is satisfied (see the line with dots).

\begin{figure}
\begin{center}
\includegraphics [height=9.5cm]{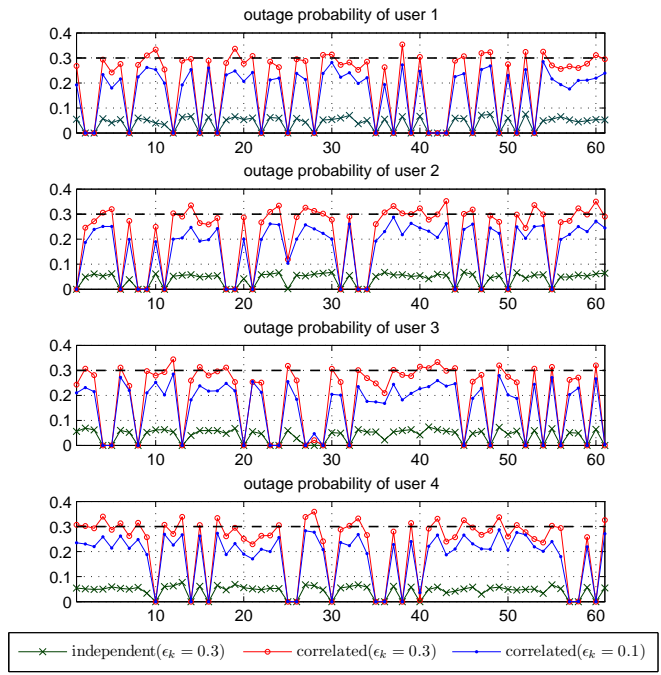}
\caption{Comparison of outage probability of $4$ users with and
without frequency correlations in channel model.}
\label{fg_freq.corr}
\end{center}
\end{figure}

\section{Conclusions}\label{conclude}
This paper proposed a slow adaptive OFDMA scheme that can achieve a
throughput close to that of fast adaptive OFDMA schemes, while
significantly reducing the computational complexity and control
signaling overhead. Our scheme can satisfy user data rate
requirement with high probability. This is achieved by formulating
our problem as a stochastic optimization problem. Based on this
formulation, we design a polynomial-time algorithm for subcarrier
allocation in slow adaptive OFDMA. Our simulation results showed
that the proposed algorithm converges within 22 iterations on
average.

In the future, it would be interesting to investigate the chance
constrained subcarrier allocation problem when frequency correlation exists, or when the channel distribution information is not perfectly known at the BS. Moreover, it is worthy to study the tightness of the Bernstein approximation. Another interesting
direction is to consider discrete data rate and exclusive subcarrier
allocation. In fact, the proposed algorithm based on cutting plane
methods can be extended to incorporate integer constraints on the
variables (see e.g., \cite{Mit:03}).

Finally, our work is an initial attempt to apply the chance
constrained programming methodology to wireless system designs. As
probabilistic constraints arise quite naturally in many wireless
communication systems due to the randomness in channel conditions,
user locations, etc., we expect that chance constrained programming
will find further applications in the design of high performance
wireless systems.

\appendices
\section{Bernstein Approximation Theorem}\label{apdx_bern}
\begin{theorem}\label{thm_bern}
Suppose that
$F(\mathbf{x},\mathbf{r}):\mathbb{R}^n\times\mathbb{R}^{n_r}\rightarrow\mathbb{R}$
is a function of $\mathbf{x}\in\mathbb{R}^n$ and
$\mathbf{r}\in\mathbb{R}^{n_r}$, and $\mathbf{r}$ is a random vector
whose components are nonnegative. For every $\epsilon>0$, if there
exists an $\mathbf{x}\in\mathbb{R}^n$ such that
\begin{equation}\label{apx1}
\inf_{\varrho>0}\left\{\Psi(\mathbf{x},\varrho)-\varrho\epsilon\right\}
\le 0,
\end{equation}
where 
$$
\Psi(\mathbf{x},\varrho)\triangleq
\varrho
\mathbb{E}\left\{\exp(\varrho^{-1}F(\mathbf{x},\mathbf{r}))\right\},
$$
then $\Pr\left\{F(\mathbf{x},\mathbf{r})>0\right\}\le \epsilon$.
\end{theorem}

\begin{proof}
\emph{(Sketch)} The proof of the above theorem is given in \cite{NemSha:06} in details. To help the readers to better understand the idea, we give an
overview of the proof here.

It is shown in \cite{NemSha:06} (see section 2.2 therein) that the
probability $\text{Pr}\{F(\mathbf{x},\mathbf{r})\ge0\}$ can be
bounded as follows:
$$
\text{Pr}\{F(\mathbf{x},\mathbf{r})>0\}
\le\mathbb{E}\left\{\psi(\varrho^{-1}F(\mathbf{x},\mathbf{r}))\right\}.
$$
Here, $\varrho>0$ is arbitrary, and
$\psi(\cdot):\mathbb{R}\rightarrow\mathbb{R}$ is a nonnegative,
nondecreasing, convex function satisfying $\psi(0)=1$ and
$\psi(z)>\psi(0)$ for any $z>0$. One such $\psi$ is the exponential
function $\psi(z)=\exp(z)$. If there exists a $\hat{\varrho}>0$ such
that
$$
\mathbb{E}\left\{\exp(\hat{\varrho}^{-1}F(\mathbf{x},\mathbf{r}))\right\}\le\epsilon,
$$
then $\text{Pr}\{F(\mathbf{x},\mathbf{r})>0\}\le\epsilon$. By
multiplying by $\hat{\varrho}>0$ on both sides, we obtain the
following sufficient condition for the chance constraint
$\Pr\left\{F(\mathbf{x},\mathbf{r})>0\right\}\le \epsilon$ to hold:
\begin{equation}\label{apx11}
\Psi(\mathbf{x},\hat{\varrho})-\hat{\varrho}\epsilon\le0.
\end{equation}
In fact, condition \eqref{apx11} is equivalent to \eqref{apx1}.
Thus, the latter provides a conservative approximation of the chance
constraint.
\end{proof}

\bibliographystyle{IEEEtran}
\bibliography
{IEEEabrv,StringDefinitions,CVStringDefinitions,BiblioCV,WGroup,reference}

\input{williamli.tex} \vspace{-0.4cm}
\input{angelazhang.tex} \vspace{-0.4cm}
\input{anthonyso.tex}
\input{moewin.tex}

\end{document}

%% file: williamli.tex
\begin{biography}
	[{\includegraphics[width=1in,height=1.25in,clip,keepaspectratio]{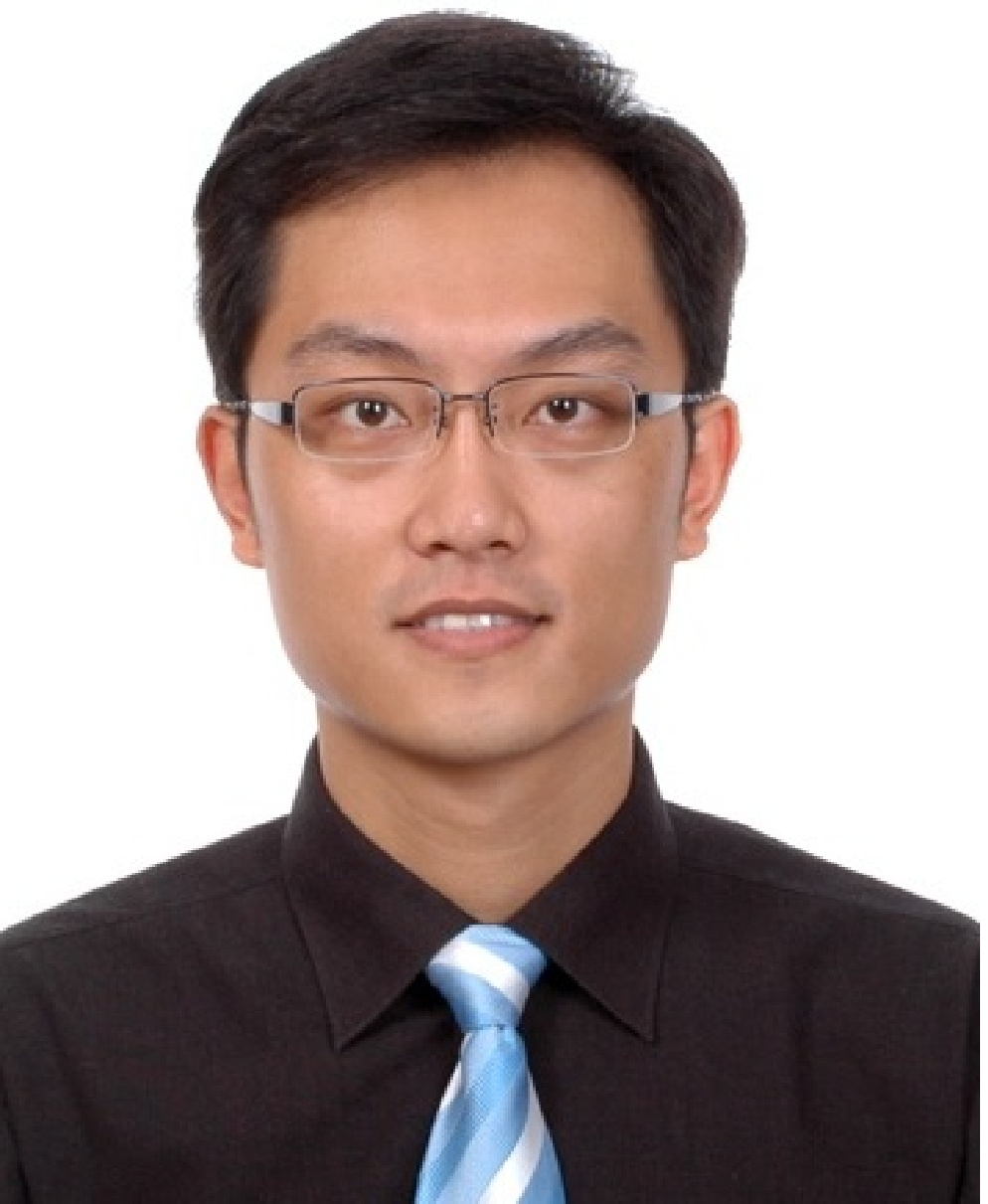}}]
	{\bf William Wei-Liang Li} (S'09) received the B.S. degree (with highest honor) in Automatic Control Engineering from Shanghai Jiao Tong University (SJTU), China in 2006. Since Aug. 2007, he has been with the Department of Information Engineering, the Chinese University of Hong Kong (CUHK), where he is now a Ph.D. candidate.
	
	From 2006 to 2007, he was with the Circuit and System Laboratory, Peking University (PKU), China, where he worked on signal processing and embedded system design. Currently, he is a visiting graduate student at the Laboratory for Information and Decision Systems (LIDS), Massachusetts Institute of Technology (MIT). His main research interests are in the wireless communications and networking, specifically broadband OFDM and multi-antenna techniques, pragmatic resource allocation algorithms and stochastic optimization in wireless systems.

	He is currently a reviewer of {\scshape IEEE Transactions on Wireless Communications}, IEEE International Conference on Communications (ICC), IEEE Consumer Communications and Networking Conference (CCNC), European Wireless and Journal of Computers and Electrical Engineering.

	During the four years of undergraduate study, he was consistently awarded the first-class scholarship, and graduated with highest honors from SJTU. He received the First Prize Award of the National Electrical and Mathematical Modelling Contest in 2005, the Award of CUHK Postgraduate Student Grants for Overseas Academic Activities and the Global Scholarship for Research Excellence from CUHK in 2009.
\end{biography}

%% file: angelazhang.tex
\begin{biography}
	[{\includegraphics[width=1in,height=1.25in,clip,keepaspectratio]{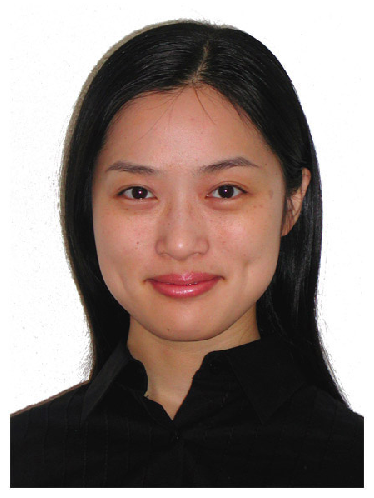}}]
	{\bf Ying Jun (Angela) Zhang} (S'00-M'05) received her Ph.D. degree in Electrical and Electronic Engineering from the Hong Kong University of Science and Technology, Hong Kong in 2004.
	
	Since Jan. 2005, she has been with the Department of Information Engineering in The Chinese University of Hong Kong, where she is currently an Assistant Professor. Her research interests include wireless communications and mobile networks, adaptive resource allocation, optimization in wireless networks, wireless LAN/MAN, broadband OFDM and multicarrier techniques, MIMO signal processing.

	Dr. Zhang is on the Editorial Boards of {\scshape IEEE Transactions on Wireless Communications} and Wiley Security and Communications Networks Journal. She has served as a TPC Co-Chair of Communication Theory Symposium of IEEE ICC 2009, Track Chair of ICCCN 2007, and Publicity Chair of IEEE MASS 2007. She has been serving as a Technical Program Committee Member for leading conferences including IEEE ICC, IEEE Globecom, IEEE WCNC, IEEE ICCCAS, IWCMC, IEEE CCNC, IEEE ITW, IEEE MASS, MSN, ChinaCom, etc. Dr. Zhang is an IEEE Technical Activity Board GOLD Representative, 2008 IEEE GOLD Technical Conference Program Leader, IEEE Communication Society GOLD Coordinator, and a Member of IEEE Communication Society Member Relations Council (MRC).
	
	As the only winner from Engineering Science, Dr. Zhang has won the Hong Kong Young Scientist Award 2006, conferred by the Hong Kong Institution of Science.
\end{biography}

%% file: anthonyso.tex
\begin{biography}
	[{\includegraphics[width=1in,height=1.25in,clip,keepaspectratio]{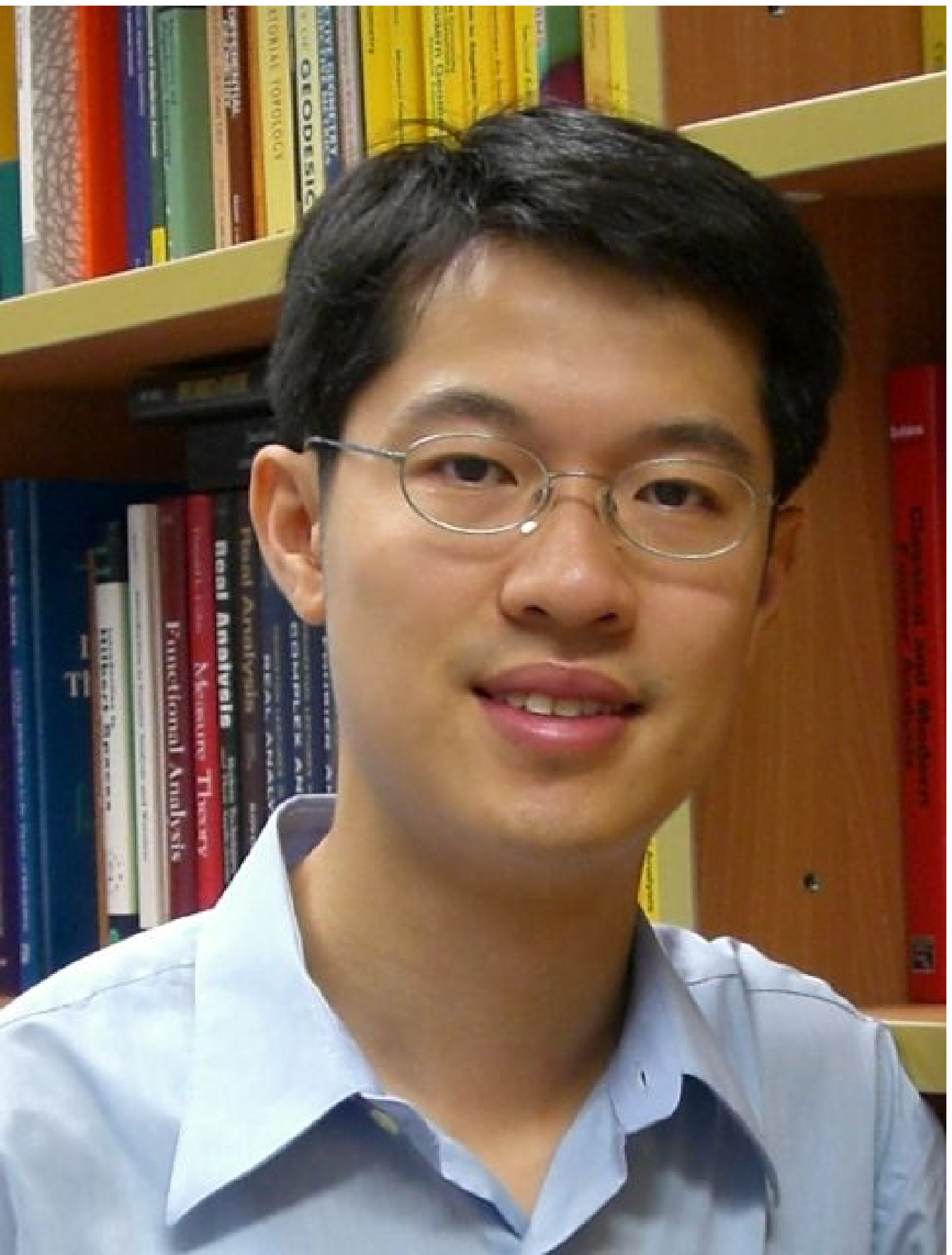}}]
	{\bf Anthony Man-Cho So} received his BSE degree in Computer Science from Princeton University in 2000 with minors in Applied and Computational Mathematics, Engineering and Management Systems, and German Language and Culture. He then received his MSc degree in Computer Science in 2002, and his Ph.D. degree in Computer Science with a Ph.D. minor in Mathematics in 2007, all from Stanford University.  
	
	Dr. So joined the Department of Systems Engineering and Engineering Management at the Chinese University of Hong Kong in 2007.  His current research focuses on the interplay between optimization theory and various areas of algorithm design, with applications in portfolio optimization, stochastic optimization, combinatorial optimization, algorithmic game theory, signal processing, and computational geometry.  
	
	Dr. So is a recipient of the 2008 Exemplary Teaching Award given by the Faculty of Engineering at the Chinese University of Hong Kong.
\end{biography}

%% file: moewin.tex
\begin{biography}
	[{\includegraphics[width=1in,height=1.25in,clip,keepaspectratio]{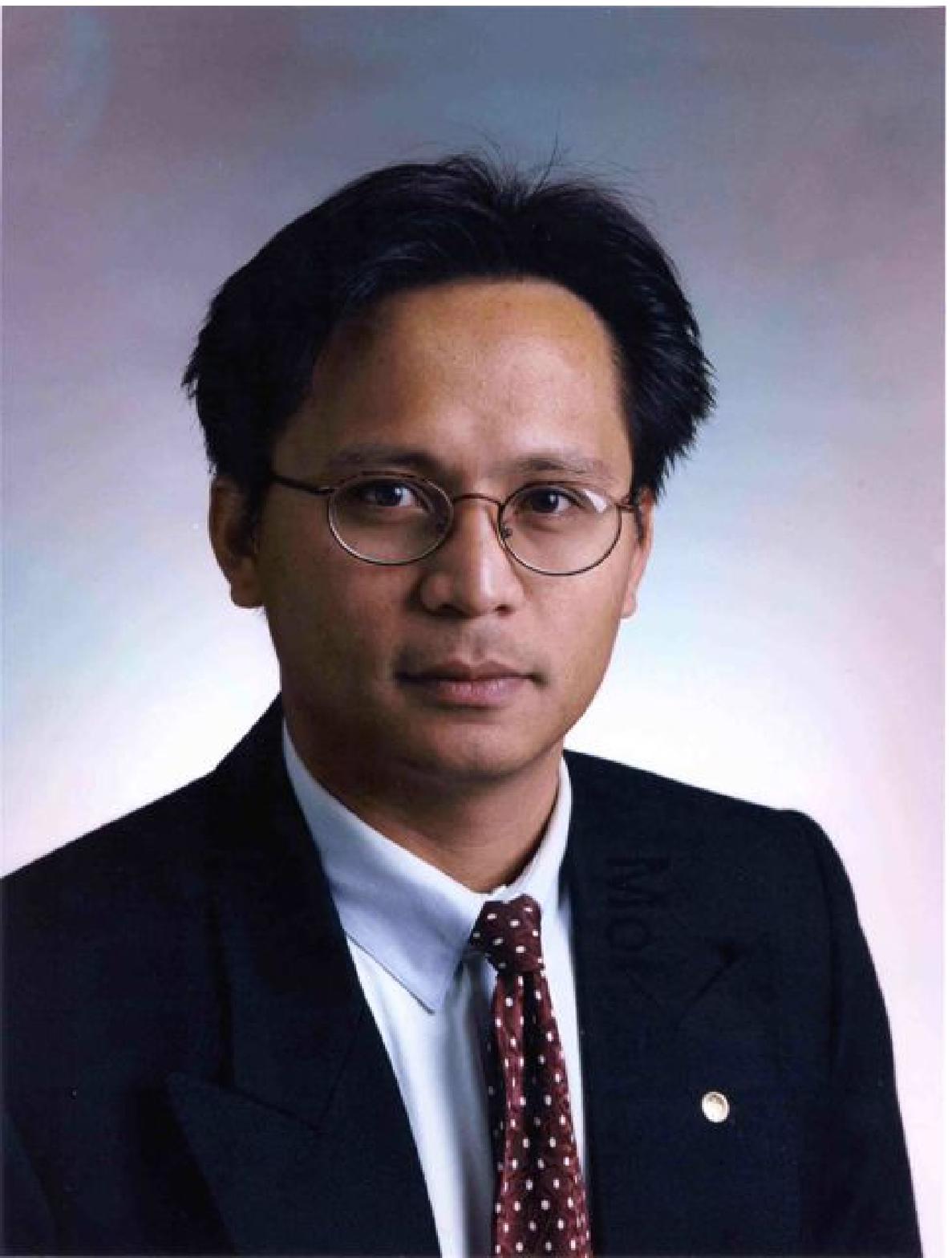}}]
	{\bf Moe Z. Win} (S'85-M'87-SM'97-F'04)
	received both the Ph.D.\ in Electrical Engineering and M.S.\ in Applied Mathematics
	as a Presidential Fellow at the University of Southern California (USC) in 1998.
	He received an M.S.\ in Electrical Engineering from USC in 1989, and a B.S.\ ({\em magna cum laude})
	in Electrical Engineering from Texas A\&M University in 1987.

	Dr.\ Win is an Associate Professor at the Massachusetts Institute of Technology (MIT).
	Prior to joining MIT, he was at AT\&T Research
	Laboratories for five years and at the Jet Propulsion Laboratory for seven years.
	His research encompasses developing fundamental theories, designing algorithms, and
	conducting experimentation for a broad range of real-world problems.
	His current research topics include location-aware networks,
	time-varying channels, multiple antenna systems, ultra-wide bandwidth
	systems, optical transmission systems, and space communications systems.

	Professor Win is an IEEE Distinguished Lecturer and
	        elected Fellow of the IEEE, cited for ``contributions to wideband wireless transmission.''
	He was honored with
	        the IEEE Eric E. Sumner Award (2006), an IEEE Technical Field Award for
	        ``pioneering contributions to ultra-wide band communications science and technology.''
	Together with students and colleagues, his papers have received several awards including
	        the IEEE Communications Society's Guglielmo Marconi Best Paper Award (2008)
	    and the IEEE Antennas and Propagation Society's Sergei A. Schelkunoff Transactions Prize Paper Award (2003).
	His other recognitions include
	        the Laurea Honoris Causa from the University of Ferrara, Italy (2008),
	        the Technical Recognition Award of the IEEE ComSoc Radio Communications Committee (2008),
	        Wireless Educator of the Year Award (2007),
	        the Fulbright Foundation Senior Scholar Lecturing and Research Fellowship (2004),
	        the U.S. Presidential Early Career Award for Scientists and Engineers (2004),
	        the AIAA Young Aerospace Engineer of the Year (2004),
	    and the Office of Naval Research Young Investigator Award (2003).

	Professor Win has been actively involved in organizing and chairing
	a number of international conferences. He served as
	    the Technical Program Chair for
	        the IEEE Wireless Communications and Networking Conference in 2009,
	        the IEEE Conference on Ultra Wideband in 2006,
	        the IEEE Communication Theory Symposia of ICC-2004 and Globecom-2000,
	        and
	        the IEEE Conference on Ultra Wideband Systems and Technologies in 2002;
	    Technical Program Vice-Chair for
	        the IEEE International Conference on Communications in 2002; and
	    the Tutorial Chair for
	        ICC-2009 and
	        the IEEE Semiannual International Vehicular Technology Conference in Fall 2001.
	He was
	    the chair (2004-2006) and secretary (2002-2004) for
	        the Radio Communications Committee of the IEEE Communications Society.
	Dr.\ Win is currently
	    an Editor for {\scshape IEEE Transactions on Wireless Communications.}
	He served as
	    Area Editor for {\em Modulation and Signal Design} (2003-2006),
	    Editor for {\em Wideband Wireless and Diversity} (2003-2006), and
	    Editor for {\em Equalization and Diversity} (1998-2003),
	        all for the {\scshape IEEE Transactions on Communications}.
	He was Guest-Editor
	        for the
	        {\scshape Proceedings of the IEEE}
	        (Special Issue on UWB Technology \& Emerging Applications) in 2009 and
	        {\scshape IEEE Journal on Selected Areas in Communications}
	        (Special Issue on Ultra\thinspace-Wideband Radio in Multiaccess
	        Wireless Communications) in 2002.
\end{biography}